\documentclass[runningheads]{llncs}

\usepackage{graphicx}
\usepackage{amsmath,amssymb,amsfonts}
\usepackage{dsfont}
\usepackage{booktabs}
\usepackage{array}
\usepackage{caption}
\usepackage{float}
\usepackage{tabularx}
\usepackage{url}
\usepackage{cite}
\usepackage{subfig}
\usepackage{hyperref}
\usepackage{subcaption}
\usepackage{xcolor}
\usepackage{algorithm}
\usepackage{algpseudocode}

\hypersetup{
colorlinks=true,
linkcolor=blue,
citecolor=blue,
urlcolor=blue,
pdfborder={0 0 0},
bookmarksnumbered=true
}

\newcommand{\F}{\mathbb{F}}

\newcommand{\wt}[1]{\mathrm{wt}(#1)}

\newcommand{\Adv}[1]{\mathsf{Adv}_{#1}}
\newcommand{\EUF}{\mathsf{EUF\text{-}CMA}}
\newcommand{\xor}{\oplus}
\newcommand{\pk}{\mathsf{pk}}
\newcommand{\sk}{\mathsf{sk}}
\newcommand{\Hpub}{H_{\mathsf{pub}}}
\newcommand{\Hsec}{H_{\mathsf{sec}}}

\newcommand{\calA}{\mathcal{A}}
\newcommand{\calB}{\mathcal{B}}
\newcommand{\calD}{\mathcal{D}}

\begin{document}

\title{Post-Quantum Sanitizable Signatures from McEliece-Based Chameleon Hashing}
\titlerunning{PQ Sanitizable Signatures from McEliece Chameleon Hashing}

%
%

	\author{Shahzad Ahmad\inst{1}\orcidID{0000-0002-9654-869X} \and
		Stefan Rass\inst{1}\orcidID{0000-0003-2821-2489} \and
		Zahra Seyedi\inst{2}\orcidID{0009-0002-8492-4640}}
	
	\authorrunning{S. Ahmad et al.}
	
	\institute{LIT Secure and Correct Systems Lab, Johannes Kepler University, Linz, Austria
		\and
		Department of Electronics, Information and Bioengineering, Polytechnic University of Milan,
		Milan, Italy\\
		\email{shahzad.ahmad@jku.at},
		\email{stefan.rass@jku.at},
		\email{zahrasseyedi@gmail.com}}


\maketitle

\begin{abstract}
We introduce a novel post-quantum sanitizable signature scheme constructed upon a chameleon hash function derived from the McEliece cryptosystem. In this design, the designated sanitizer possesses the inherent trapdoor of a Goppa code, which facilitates controlled collision-finding via Patterson decoding. This mechanism enables authorized modification of specific message blocks while ensuring all other content remains immutably bound. We provide formal security definitions and rigorous proofs of existential unforgeability and immutability, grounded in the hardness of syndrome decoding in the random-oracle model, where a robust random oracle thwarts trivial linear hash collisions. A key innovation lies in our precise characterization of the transparency property: by imposing a specific weight constraint on the randomizers generated by the signer, we achieve perfect transparency, rendering sanitized signatures indistinguishable from freshly signed ones. This work establishes the first transparent, code-based, post-quantum sanitizable signature scheme, offering strong theoretical guarantees and a pathway for practical deployment in long-term secure applications.

\end{abstract}


\section{Introduction}\label{sec:intro}
A sanitizable signature scheme allows a signer to partially delegate modification rights to a designated \emph{sanitizer}: selected blocks of a signed document can be altered by the sanitizer without invalidating the signature, while all remaining blocks stay immutable~\cite{ateniese_sanitizable_2005}. The signer specifies an \emph{admissibility mask} $\mathit{adm}$ that indicates which blocks are mutable.
The sanitizer, who holds secret trapdoor information, can update those blocks and adjust the signature accordingly. Canonical applications include redacting patient identifiers from physician-signed medical records, updating expiration dates in X.509 certificates, and anonymizing supply chain documentation.

The central privacy property is \emph{transparency}: an observer, given only the sanitized document and its signature, should not be able to determine whether the signature originated with the original signer or the sanitizer~\cite{ateniese_sanitizable_2005}. Formally, the distributions of fresh and sanitized signatures should be (statistically or computationally) close. Complementary properties are \emph{immutability}, which ensures the sanitizer cannot alter blocks outside the mask, and \emph{unforgeability}, which ensures no third party can produce valid signatures.

\paragraph{The quantum threat.}
Sanitizable signatures are foundational tools for regulated data modification and have always depended on trapdoor hashes derived from RSA, discrete logarithms, or pairings~\cite{krawczyk_chameleon_1998,miyazaki_digitally_2006,derler_fine-grained_2019}. Yet, all of these primitives are definitively broken in polynomial time by Shor's algorithm~\cite{shor_algorithms_1994} on a sufficiently large quantum computer. In direct response to this threat, lattice-based sanitizable signatures have very recently emerged~\cite{ClermontPQ2025}, but they strictly require the random-oracle model (ROM) for their chameleon hash component. Until now, code-based sanitizable signatures have not been realised.

We decisively fill this gap by constructing the first transparent, post-quantum, sanitizable signature scheme based on code-based assumptions. A naive linear hash function $H(x,r)=(x \xor r)\cdot \Hpub^T$ is unambiguously vulnerable to a trivial collision attack: an adversary who sees $(x,r)$ can freely choose $r'\ne r$ and set $x'=x \xor r \xor r'$ to obtain a collision without the trapdoor. Our key insight preempts this attack by preprocessing the message input with a random oracle, defining the hash as $H(m,r)=(G(m)\xor r)\cdot\Hpub^T$, where $G$ is modelled as SHA-3. Finding a collision now categorically requires inverting $G$, which is infeasible.

We provide a precise transparency analysis. Constraining the signer to sample randomizers of exactly weight $t$ (matching Patterson decoding's output distribution) achieves perfect transparency ($\delta=0$) for any number of admissible blocks. If the signer samples randomizers of weight \emph{at most} $t$, the statistical distance is sharply bounded: $\delta = 1 - \binom{n}{t}/\!\sum_{j=0}^{t}\binom{n}{j} \approx 0.0187$ for standard parameters ($n=3488$, $t=64$).

\subsection{Our Contributions}

\begin{enumerate}
\item \textit{ROM-secure McEliece chameleon hash.}
We define $H_\pk(m,r) = (G(m)\xor r)\cdot\Hpub^T$, where $G$ is a SHA-3 random oracle.
Collision resistance follows from ROM + $\mathrm{SD}(n,n{-}k,2t)$ hardness: a collision yields a weight-$\le 2t$ vector whose syndrome under $\Hpub$ equals a non-zero ROM-derived target, which is a genuine SD instance. The trapdoor (Goppa-code secret key) enables Patterson decoding to find collisions in $O(n\cdot t)$ operations.

\item \textit{Sanitizable signature scheme.}
We build a complete scheme using two independent McEliece chameleon hashes: one for immutable blocks and one for admissible blocks, chained across message blocks, with a Dilithium2-based signature on the final digest. Unforgeability and immutability are proved in the ROM.

\item \textit{Tight transparency analysis.}
We show that constraining the signer to sample weight-$t$ randomizers achieves
perfect transparency ($\delta=0$). We also derive the non-zero distance
$\delta \approx 0.0187$ for the unconstrained case and prove a hybrid-argument bound of $\mathsf{Adv}\le L\cdot\delta + \varepsilon_{\mathsf{sig}}$.
This provides a clear trade-off between implementation choices and privacy guarantees.

\item \textit{Implementation and empirical benchmarks.}
A Python prototype on NIST Category~1 Classic-McEliece parameters ($n=3488$, $k=2720$, $t=64$) validates all theoretical claims. Theoretical Patterson time is $\approx 8$\,ms per modified block. Public key size is $\approx 655.3$\,KB; signature size is $\approx 6.75$\,KB for $L=10$.
\end{enumerate}

\section{Related Work}\label{sec:relatedwork}

\paragraph{Sanitizable signatures.}
The primitive was introduced by Ateniese et al.~\cite{ateniese_sanitizable_2005}, who formalised immutability and transparency using RSA and discrete-log trapdoor hashes. Brzuska et al.~\cite{jarecki_security_2009} strengthened the security model established by Ateniese et al. by clarifying the roles of signing and sanitization oracles and by extending the formal treatment of accountability. Content-extraction signatures~\cite{steinfeld_content_2002} and redactable signatures are closely related, but differ from sanitizable signatures: they either assume a different trust model or do not use a trapdoor for sanitization.

\paragraph{Chameleon hashes.}
Krawczyk and Rabin~\cite{krawczyk_chameleon_1998} formalized chameleon hashes, followed by classical constructions using RSA~\cite{rivest_method_nodate}, discrete logarithms~\cite{elgamal_public_1985}, or pairings~\cite{miyazaki_digitally_2006}. Subsequently, Derler et al.~\cite{derler_fine-grained_2019} introduced fine-grained and attribute-based variants. Bellare and Ristov~\cite{bellare_hash_2008} generalised the construction by building chameleon hashes from $\Sigma$-protocols; for instance, incorporating code-based identification (e.g., Stern's protocol~\cite{stern1993new} or the CROSS-based protocol~\cite{CROSS2024}) produces a ROM-based code chameleon hash. In contrast, our construction uses the McEliece syndrome structure directly, making the connection to Patterson decoding explicit and enabling parameters to be compared directly to those of Classic-McEliece.

\paragraph{Collision-resistance notions.}
Derler, Samelin, and Slamanig~\cite{derler_collision_2020} survey the hierarchy of collision-resistance notions for chameleon hashes, including collision resistance (CR), second-preimage resistance (SPR), and message-binding collision resistance (CMR). Our construction satisfies CR and CMR in the ROM; exposure-freeness is not claimed, and we discuss this limitation in Section~\ref{sec:discussion}.

\paragraph{Post-quantum sanitizable signatures.}
Until very recently, no post-quantum sanitizable signature existed. Clermont et al.~\cite{ClermontPQ2025} present the first lattice-based construction, instantiated in the ROM with structured lattice trapdoors and verifiable ring signatures. Their scheme achieves post-quantum security but relies on the ROM for chameleon hash collision resistance, as does ours. Neither construction has been shown to be end-to-end standard-model secure. Our scheme uses a fundamentally different assumption (syndrome decoding vs \ Module-LWE/SIS), relies on a 45-year-old hard problem, and achieves \emph{perfect transparency} via the weight-$t$ constraint.

\paragraph{Code-based hash functions and signatures.}
Augot, Finiasz, and Sendrier~\cite{augot_finiasz_2005} construct fast syndrome-based hash functions and discuss trapdoors, though not in the chameleon-hash framework.
Courtois, Finiasz, and Sendrier~\cite{cryptoeprint:2001/010} exploit Patterson decoding as a signing trapdoor. Our work repurposes syndrome decoding for \emph{controlled collision finding} in a chameleon hash, which is a conceptually distinct application. Table~\ref{tab:comparison} positions our scheme among existing constructions.

\begin{table}[t]
\centering
\caption{Comparison of sanitizable signature schemes. PQ = Post-Quantum;
	CH Model = security model for chameleon hash collision resistance.}
\label{tab:comparison}
\begin{tabularx}{\textwidth}{@{}lXccXl@{}}
	\toprule
	\textbf{Scheme} & \textbf{Trapdoor} & \textbf{PQ} &
	\textbf{CH Model} & \textbf{Transp.} & \textbf{Collision Mech.}\\
	\midrule
	Ateniese et al.~\cite{ateniese_sanitizable_2005} & RSA/DL & No &
	Standard & Comput. & Factoring/DL trapdoor \\
	Miyazaki et al.~\cite{miyazaki_digitally_2006} & Pairings & No &
	Standard & Strong & DDH collision \\
	Clermont et al.~\cite{ClermontPQ2025} & Lattice & Yes &
	ROM & Weak & Preimage sampling \\
	\midrule
	\textbf{This work} & \textbf{McEliece} & \textbf{Yes} &
	\textbf{ROM} & \textbf{Perfect ($\delta=0$, w/ constraint)} &
	\textbf{Patterson decoding} \\
	\bottomrule
\end{tabularx}
\end{table}

\section{Background and Formal Definitions}\label{sec:background}

\subsection{Notation}
We write $[n]=\{1,\dots,n\}$, $\kappa$ for the security parameter, PPT for probabilistic polynomial time, and $\mathsf{negl}(\kappa)$ for a negligible function. All vectors are row vectors over $\F_2$. $\wt{v}$ denotes the Hamming weight of $v$. For a binary Goppa code $C[n,k,2t+1]$, we denote its secret parity-check matrix $\Hsec\in\F_2^{(n-k)\times n}$ and public (masked) matrix $\Hpub\in\F_2^{(n-k)\times n}$.

\subsection{Computational Assumptions}\label{subsec:assumptions}

\begin{definition}[Syndrome Decoding (SD)]
Let $H\in\F_2^{\rho\times n}$, $s\in\F_2^{\rho}$, $w\in\mathbb{Z}_{>0}$. The \emph{syndrome decoding problem} $\mathrm{SD}(n,\rho,w)$ asks to find $e\in\F_2^n$ with $\wt{e}\le w$ and $e\cdot H^T=s$. (Here $\rho$ denotes the number of parity-check rows; throughout the paper $r$ is reserved for chameleon-hash randomizer vectors.) $\mathrm{SD}$ is NP-complete in the worst case~\cite{berlekamp_inherent_1978} and is believed to be hard on average for random $H$. The best quantum algorithms for $\mathrm{SD}$ (based on quantum ISD) provide only a sub-exponential speedup, leaving 128-bit post-quantum security intact
for Classic-McEliece-3488.
\end{definition}

\begin{definition}[McEliece Assumption]
The \emph{McEliece decisional assumption} states that a masked parity-check matrix $\Hpub = S'\cdot\Hsec\cdot P$ (with random invertible $S'$ and permutation $P$) is computationally indistinguishable from a uniformly random binary matrix, for any PPT distinguisher given $\Hpub$ but not $(S',P)$.
\end{definition}

\subsection{Chameleon Hash Functions}\label{subsec:ch}

\begin{definition}[Chameleon Hash~\cite{krawczyk_chameleon_1998}]
A chameleon hash scheme \[\mathsf{CH}=(\mathsf{CHGen},\mathsf{CH},\mathsf{CHCol})\]
consists of:
\begin{itemize}
	\item $\mathsf{CHGen}(1^\kappa)\to(\pk_\mathsf{ch},\sk_\mathsf{ch})$:
	generates a public hashing key and secret trapdoor.
	\item $\mathsf{CH}_{\pk}(m,r)\to h$: hashes message $m$ with randomness
	$r$ to digest $h$.
	\item $\mathsf{CHCol}_{\sk}(m,r,m')\to r'$: given the trapdoor, finds
	$r'$ such that $\mathsf{CH}_\pk(m,r)=\mathsf{CH}_\pk(m',r')$.
\end{itemize}
\textbf{Collision resistance (CR):} No PPT adversary without $\sk$ can find
$(m,r)\ne(m',r')$ with $\mathsf{CH}_\pk(m,r)=\mathsf{CH}_\pk(m',r')$ except
with negligible probability.
\end{definition}

\subsection{Sanitizable Signature Schemes}\label{subsec:sss}

\begin{definition}[Sanitizable Signature Scheme~\cite{ateniese_sanitizable_2005,jarecki_security_2009}]
A sanitizable signature scheme $\mathsf{SSS}=
(\mathsf{KeyGen},\mathsf{Sign},\mathsf{Sanitize},\mathsf{Verify})$ operates as follows:
\begin{itemize}
	\item $\mathsf{KeyGen}(1^\kappa)\to
	((\pk_\mathsf{sig},\sk_\mathsf{sig}),(\pk_\mathsf{san},\sk_\mathsf{san}))$:
	generates key pairs for signer and sanitizer.
	\item $\mathsf{Sign}(\sk_\mathsf{sig},\pk_\mathsf{san},M,\mathit{adm})\to\sigma$:
	the signer produces a signature on message $M=(M[1],\dots,M[L])$ and mask
	$\mathit{adm}\in\{0,1\}^L$, where $\mathit{adm}[i]=1$ marks block $i$
	as admissible.
	\item $\mathsf{Sanitize}(\sk_\mathsf{san},M,\sigma,M')\to\sigma'$:
	given a valid signature on $(M,\mathit{adm})$ and a message $M'$ that
	differs from $M$ only on admissible blocks, the sanitizer outputs a
	valid signature on $(M',\mathit{adm})$.
	\item $\mathsf{Verify}(\pk_\mathsf{sig},\pk_\mathsf{san},M,\sigma)\to
	\{\mathsf{True},\mathsf{False}\}$.
\end{itemize}
\end{definition}

\paragraph{Security properties.}
We recall the three core security games:

\begin{definition}[$\EUF$-Security~\cite{jarecki_security_2009}]
An adversary $\calA$ wins the EUF game if it outputs a valid signature $(M^*,\sigma^*)$ on a message $M^*$ not previously queried to the signing oracle (or sanitization oracle producing the same authenticated content). The scheme is $\EUF$-secure if no PPT adversary wins with non-negligible advantage.
\end{definition}

\begin{definition}[Immutability~\cite{ateniese_sanitizable_2005}]
An adversary $\calA$ (who may hold $\sk_\mathsf{san}$) wins the immutability game if it produces a valid signature $(M',\sigma')$ where $M'$ differs from any previously signed $M$ in some block $i$ with $\mathit{adm}[i]=0$. The scheme is immutable if no PPT adversary wins with non-negligible advantage.
\end{definition}

\begin{definition}[Transparency~\cite{ateniese_sanitizable_2005,brzuska_sanitizable}]
A scheme achieves \emph{strong transparency} if, for any PPT distinguisher $\calA$ given either a fresh signature (by signer) or a sanitized signature (by sanitizer) on the same final message, $\calA$'s advantage in distinguishing them is negligible.
It achieves \emph{weak transparency} if the distinguishing advantage is bounded by some $\epsilon<1/2$. We write $\Adv{\calA}^\mathsf{trans}$ for $\calA$'s distinguishing advantage in the transparency game.
\end{definition}

\begin{remark}[Accountability vs.\ Transparency]
Some models include an \emph{accountability} property allowing an external judge to determine whether a signature was produced by the signer or sanitizer. Our scheme prioritizes transparency (privacy): sanitized signatures are indistinguishable from fresh ones. Accountability is therefore not provided. However, the signer retains the original message and randomness as private evidence, providing a form of signer-side accountability.
\end{remark}

\section{Construction}\label{sec:construction}

\subsection{McEliece-Based Chameleon Hash (ROM)}\label{subsec:ch-construction}
Let $C\subset\F_2^n$ be a binary Goppa code with parameters $[n,k,2t+1]$ and secret parity-check matrix $\Hsec\in\F_2^{(n-k)\times n}$. Let $G:\{0,1\}^*\to\F_2^n$ be a hash function modelled as a random oracle (instantiated with SHA-3 in counter mode).

\paragraph{Key generation.}
Sample random invertible $S'\in\F_2^{(n-k)\times(n-k)}$ and permutation $P\in\F_2^{n\times n}$. Set $\Hpub = S'\cdot\Hsec\cdot P$. The public key is $\pk_\mathsf{ch}=\Hpub$; the secret key is $\sk_\mathsf{ch}=(P,(S')^{-1},\text{Goppa parameters})$. (We store $P$ directly because the collision-finding algorithm computes
$r'=f\cdot P$; since $P$ is a permutation matrix, $P^{-1}=P^T$ can be recovered trivially when needed.)

\paragraph{Hash function.}
For $m\in\F_2^n$ and $r\in\F_2^n$ with $\wt{r}\le t$:
\[
\mathsf{CH}_\pk(m,r) = \bigl(G(m) \xor r\bigr)\cdot\Hpub^T
\;\in\F_2^{n-k}.
\]

\paragraph{Why ROM preprocessing is essential.}
Without $G$, the hash is $H(x,r)=(x\xor r)\cdot\Hpub^T$. An adversary observing $(x,r)$ in a public signature could trivially find a collision: pick any $r'\ne r$ with $\wt{r'}\le t$, set $x'=x\xor r\xor r'$, and check that $(x'\xor r')\cdot\Hpub^T = (x\xor r)\cdot\Hpub^T$. This attack exploits the linearity of the hash and requires no trapdoor. With $G$ modelled as a ROM, the adversary must invert $G$ to choose $m'$
such that $G(m')\xor r'$ matches the required value, which is infeasible.

\paragraph{Collision finding with trapdoor.}
Given $(m,r,m')$, the sanitizer computes:
\[
s_\mathsf{target}
= \bigl(G(m)\xor r \xor G(m')\bigr)\cdot\Hpub^T
\;\in\F_2^{n-k}.
\]
This is the syndrome that $r'$ must satisfy: $r'\cdot\Hpub^T=s_\mathsf{target}$
with $\wt{r'}\le t$.
Using $\sk_\mathsf{ch}=\bigl(P,\,(S')^{-1},\,\text{Goppa parameters}\bigr)$:
\begin{enumerate}
	\item Compute $s_{\mathsf{pp}}=s_{\mathsf{target}}\cdot(S')^{-T}\in\F_2^{n-k}$.
	($(S')^{-T}$ is the transpose of the inverse of the $(n{-}k)\times(n{-}k)$ scrambler $S'$.)
	\item Apply Patterson decoding to find $f\in\F_2^n$ with $f\cdot\Hsec^T = s_{\mathsf{pp}}$ and $\wt{f}\le t$.
	\item Recover $r'$ via the column permutation: set $r'=f\cdot P$,
	where $P\in\F_2^{n\times n}$ is the permutation matrix stored in
	$\sk_\mathsf{ch}$, satisfying $\Hpub=S'\cdot\Hsec\cdot P$.
\end{enumerate}
\textbf{Correctness.}
Since $\Hpub^T = P^T\cdot\Hsec^T\cdot(S')^T$,
\[
r'\cdot\Hpub^T
= (f\cdot P)\cdot P^T\cdot\Hsec^T\cdot(S')^T
= f\cdot\Hsec^T\cdot(S')^T
= s_{\mathsf{pp}}\cdot(S')^T
= s_{\mathsf{target}},
\]
confirming $\mathsf{CH}_\pk(m,r)=\mathsf{CH}_\pk(m',r')$. Patterson decoding always succeeds when $\wt{f}\le t$ and runs in $O(n^2 t)$ bit operations (practically $O(n\cdot t)$ for the dominant polynomial operations).

\begin{remark}[Goppa code requirements]
For the Goppa code to support Patterson decoding, its defining polynomial $g(x)$ must be irreducible over $\F_{2^m}$ with no repeated roots in the support $L_\mathsf{Goppa}$, ensuring minimum distance $\ge 2t+1$. These constraints are standard in Classic-McEliece and are enforced at key generation.
\end{remark}

\begin{remark}[Weight-$t$ exactness of decoded error]\label{rem:weightt}
The verifier accepts randomizers with weight up to $t$, but Patterson decoding returns the \emph{unique} error vector $f$ with $\wt{f}\le t$. For Classic-McEliece ($n=3488$, $t=64$), the ratio $\binom{n}{t}/\sum_{j=0}^t\binom{n}{j}\approx 0.981$ indicates that nearly all correctable syndromes correspond to an error of weight exactly $t$. In the rare event that $\wt{f}<t$, our security properties hold. However,
as analysed in Section~\ref{sec:transparency}, this small distributional
difference is observable. To achieve perfect transparency, the sanitizer can
resample if $\wt{f}<t$, or the signer can constrain sampling to weight exactly $t$.
\end{remark}

\subsection{Sanitizable Signature Scheme}\label{subsec:sss-construction}

\paragraph{Parameters and message format.}
Fix code parameters $(n,k,t)$ providing $\kappa$-bit post-quantum security.
Messages are split into $L$ blocks
$M=(M[1],\dots,M[L])$ with $M[i]\in\F_2^k$.
For each block, the previous $(n-k)$-bit hash is concatenated with
$M[i]$ to form an $n$-bit input to the chameleon hash.

\vspace{0.5em}

\begin{center}
	\fbox{
		\begin{minipage}{0.96\linewidth}
			\textbf{Key Generation $\mathsf{KeyGen}(1^\kappa)$}
			\begin{enumerate}
				\item Generate two independent binary Goppa codes
				$C_\mathsf{non},C_\mathsf{san}$ with parameters $[n,k,2t+1]$ and
				parity-check matrices $\Hsec^\mathsf{non},\Hsec^\mathsf{san}$.
				\item Mask both matrices:
				\[
				\Hpub^\mathsf{non}=S'_\mathsf{non}\Hsec^\mathsf{non}P_\mathsf{non},\qquad
				\Hpub^\mathsf{san}=S'_\mathsf{san}\Hsec^\mathsf{san}P_\mathsf{san}.
				\]
				\item Run the outer signature key generation
				\[
				(\pk_\mathsf{sig},\sk_\mathsf{sig})
				\leftarrow \mathsf{KeyGen}_\mathsf{sig}(1^\kappa).
				\]
				\item Output
				\[
				\pk=(\pk_\mathsf{sig},\Hpub^\mathsf{non},\Hpub^\mathsf{san}),
				\]
				\[
				\sk_\mathsf{sig},
				\qquad
				\sk_\mathsf{san}=
				(P_\mathsf{san},(S'_\mathsf{san})^{-1},
				\text{Goppa parameters of }C_\mathsf{san}).
				\]
			\end{enumerate}
	\end{minipage}}
\end{center}

\vspace{0.5em}

\begin{center}
	\fbox{
		\begin{minipage}{0.96\linewidth}
			\textbf{Signing $\mathsf{Sign}(\sk_\mathsf{sig},\pk,M,\mathit{adm})$}
			\begin{enumerate}
				\item Set $h_0 \gets 0^{n-k}$.
				\item For $i=1,\dots,L$:
				\begin{enumerate}
					\item $x_i \gets h_{i-1}\|M[i] \in \F_2^n$
					\item Sample $r_i \leftarrow \{e\in\F_2^n : \wt(e)=t\}$
					\item Compute
					\[
					h_i=
					\begin{cases}
						\mathsf{CH}_{\Hpub^\mathsf{non}}(x_i,r_i) & \mathit{adm}[i]=0\\
						\mathsf{CH}_{\Hpub^\mathsf{san}}(x_i,r_i) & \mathit{adm}[i]=1
					\end{cases}
					\]
				\end{enumerate}
				\item $\sigma_\mathsf{sig}
				\leftarrow
				\mathsf{Sign}_\mathsf{sig}(\sk_\mathsf{sig},\,h_L\|\mathit{adm})$
				\item Output
				\[
				\sigma=(h_L,\sigma_\mathsf{sig},(r_1,\dots,r_L),\mathit{adm}).
				\]
			\end{enumerate}
	\end{minipage}}
\end{center}

\vspace{0.5em}

\begin{center}
	\fbox{
		\begin{minipage}{0.96\linewidth}
			\textbf{Verification $\mathsf{Verify}(\pk,M,\sigma)$}
			\begin{enumerate}
				\item Parse
				$\sigma=(h'_L,\sigma_\mathsf{sig},(r_1,\dots,r_L),\mathit{adm})$
				\item Reject if any $\wt(r_i)\neq t$
				\item Set $h_0 \gets 0^{n-k}$
				\item For $i=1,\dots,L$:
				\[
				h_i=
				\begin{cases}
					\mathsf{CH}_{\Hpub^\mathsf{non}}(h_{i-1}\|M[i],r_i) & \mathit{adm}[i]=0\\
					\mathsf{CH}_{\Hpub^\mathsf{san}}(h_{i-1}\|M[i],r_i) & \mathit{adm}[i]=1
				\end{cases}
				\]
				\item Accept iff
				\[
				h_L=h'_L
				\land
				\mathsf{Verify}_\mathsf{sig}(\pk_\mathsf{sig},\,h_L\|\mathit{adm},\,\sigma_\mathsf{sig}).
				\]
			\end{enumerate}
	\end{minipage}}
\end{center}

\vspace{0.5em}

\begin{center}
	\fbox{
		\begin{minipage}{0.96\linewidth}
			\textbf{Sanitization $\mathsf{Sanitize}(\sk_\mathsf{san},M,\sigma,M')$}
			\begin{enumerate}
				\item Verify $\sigma$ on $(M,\mathit{adm})$; abort if invalid
				\item Recompute hash chain $h_0,\dots,h_L$
				\item Let
				\[
				I=\{\,i\in[L] : M[i]\neq M'[i]\land \mathit{adm}[i]=1\,\}
				\]
				(processed in ascending order)
				\item For each $i\in I$:
				\begin{enumerate}
					\item $m_\text{old}\gets h_{i-1}\|M[i]$
					\item $m_\text{new}\gets h_{i-1}\|M'[i]$
					\item Using the trapdoor of $\Hpub^\mathsf{san}$,
					find $r'_i$ such that
					\[
					\mathsf{CH}_{\Hpub^\mathsf{san}}(m_\text{old},r_i)
					=
					\mathsf{CH}_{\Hpub^\mathsf{san}}(m_\text{new},r'_i),
					\quad
					\wt(r'_i)=t
					\]
					\item Set $r_i \gets r'_i$
					\item Recompute $h_j$ for all $j\ge i$
				\end{enumerate}
				\item Output
				\[
				\sigma'=(h_L,\sigma_\mathsf{sig},(r_1,\dots,r_L),\mathit{adm})
				\]
			\end{enumerate}
	\end{minipage}}
\end{center}

\paragraph{Transparency mechanism.}
The outer signature binds only $(h_L,\mathit{adm})$.
Since sanitization preserves $h_L$ and $\mathit{adm}$,
the value $\sigma_\mathsf{sig}$ remains unchanged.
Hence sanitized and original signatures are
indistinguishable to any verifier without the trapdoor.

\section{Security Analysis}\label{sec:security}
We prove security in the random oracle model, where $G$ is an ideal hash function with no structure. All proofs are by reduction to standard hard problems.

\subsection{Collision Resistance of the Chameleon Hash (ROM)}\label{subsec:cr-proof}

\begin{theorem}[CR of $\mathsf{CH}_\pk$ in the ROM]\label{thm:cr}
Let $\calA$ be a PPT adversary making at most $q_G$ queries to $G$. If $\calA$ finds a collision $(m,r)\ne(m',r')$ with $\mathsf{CH}_\pk(m,r)=\mathsf{CH}_\pk(m',r')$ with probability $\epsilon$, then there exists a PPT algorithm $\calB$ solving $\mathrm{SD}(n,n-k,2t)$ with probability at least $\epsilon - q_G^2 / 2^{n-k}$.
\end{theorem}

\begin{proof}
Suppose $\calA$ outputs $(m,r)\ne(m',r')$ with equal hash values.
Expanding the hash definition:
\[
\bigl(G(m)\xor r\bigr)\cdot\Hpub^T
= \bigl(G(m')\xor r'\bigr)\cdot\Hpub^T,
\]
which rearranges to:
\[
\bigl(G(m)\xor r \xor G(m')\xor r'\bigr)\cdot\Hpub^T = \mathbf{0}.
\]
There are two cases.

\textbf{Case 1}: $m=m'$.
Then $G(m)=G(m')$, so $(r\xor r')\cdot\Hpub^T=\mathbf{0}$. For a random-like
$\Hpub$, this implies $r\xor r'$ is a low-weight codeword. Given $\wt{r\xor r'}\le 2t$
and minimum distance $2t+1$, this implies $r\xor r' = \mathbf{0}$, so $r=r'$.
Since $(m,r)\ne(m',r')$ and $m=m'$, we must have $r\ne r'$, a contradiction.

\textbf{Case 2}: $m\ne m'$.
Expanding the collision condition and rearranging:
\[
\bigl(r\xor r'\bigr)\cdot\Hpub^T
= \bigl(G(m)\xor G(m')\bigr)\cdot\Hpub^T.
\]
Define $e=r\xor r'\in\F_2^n$ and
$s_u=\bigl(G(m)\xor G(m')\bigr)\cdot\Hpub^T\in\F_2^{n-k}$. Then $e\cdot\Hpub^T = s_u$ and $\wt{e}\le\wt{r}+\wt{r'}\le 2t$. \emph{$s_u$ is non-zero with overwhelming probability:} since $m\ne m'$, $G(m)$ and $G(m')$ are independent uniform ROM outputs,
so $G(m)\xor G(m')$ is a uniformly random non-zero $n$-bit vector, and $s_u\ne\mathbf{0}$ except with probability $2^{-(n-k)}$.

Algorithm $\calB$ runs $\calA$ in the random oracle model, programming $G$ to record all queries. When $\calA$ outputs a collision $(m,r)\ne(m',r')$ with $m\ne m'$, $\calB$ computes $s_u$ (using the recorded ROM outputs for $m$ and $m'$), sets $e=r\xor r'$, and outputs $(\Hpub, s_u, 2t, e)$ as a solution to $\mathrm{SD}(n,n-k,2t)$: a vector of weight $\le 2t$ whose syndrome under $\Hpub$ equals the non-zero target $s_u$. This is a \emph{genuine} syndrome-decoding instance (non-zero target), not a minimum-distance instance. By the birthday bound, two ROM query outputs collide with probability $q_G^2/2^{n-k}$, which is negligible. Subtracting this probability, $\calB$ succeeds with advantage $\ge\epsilon - q_G^2/2^{n-k}$.\qed
\end{proof}

\subsection{Existential Unforgeability (EUF-CMA)}\label{subsec:euf-proof}

\begin{theorem}[EUF-CMA in the ROM]\label{thm:euf}
If the base signature scheme is $\EUF$-secure with advantage $\varepsilon_\mathsf{sig}$ and $\mathsf{CH}_\pk$ is collision-resistant with advantage $\varepsilon_\mathsf{cr}$ (Theorem~\ref{thm:cr}), then our sanitizable signature scheme achieves $\EUF$ security with advantage at most $\varepsilon_\mathsf{sig}+\varepsilon_\mathsf{cr}$.
\end{theorem}

\begin{proof}
Let $\calA$ be a PPT adversary given $\pk$ and access to signing and
sanitization oracles.
Suppose $\calA$ outputs a forged valid signature $(M^*,\sigma^*)$ with
$\sigma^*=(h^*_L,\sigma^*_\mathsf{sig},(r^*_1,\dots,r^*_L),\mathit{adm}^*)$.

\textbf{Case~1 (Fresh forgery)}: $M^*$ was never queried to the signing
oracle with the same mask $\mathit{adm}^*$. The hash chain for $M^*$ under $\mathit{adm}^*$ produces some $h^*_L$, and $\sigma^*_\mathsf{sig}$ is a valid base signature on $h^*_L\|\mathit{adm}^*$. Since this value was never presented to the signing oracle, this is a valid forgery for the base $\EUF$ scheme. Reduction $\calB_1$ simulates $\calA$'s signing oracle by querying its own base signing oracle, and forwards the forgery.

\textbf{Case~2 (Collision on immutable block)}: The forged $M^*$ agrees with some queried $M$ on mask $\mathit{adm}=\mathit{adm}^*$, but differs on some immutable block $j$ with $\mathit{adm}[j]=0$. For $\sigma^*$ to verify, the reconstructed hash chain must yield $h^*_L$ matching the signed value. At block $j$ the hash is computed with $H_\mathsf{non}$:
\[
\mathsf{CH}_{\Hpub^\mathsf{non}}(h_{j-1}\|M[j], r_j)
= \mathsf{CH}_{\Hpub^\mathsf{non}}(h_{j-1}\|M^*[j], r^*_j).
\]
Since $M[j]\ne M^*[j]$, this is a non-trivial collision on $H_\mathsf{non}$ (for which $\calA$ holds no trapdoor). Reduction $\calB_2$ extracts this collision to solve
$\mathrm{SD}(n,n{-}k,2t)$ as in Theorem~\ref{thm:cr}: it computes $e = r_j \xor r^*_j$ and $s_u = (G(h_{j-1}\|M[j])\xor G(h_{j-1}\|M^*[j]))\cdot(\Hpub^\mathsf{non})^T$, then outputs $(\Hpub^\mathsf{non},\,s_u,\,2t,\,e)$ as an SD witness.

Combining, the total forging advantage is at most
$\varepsilon_\mathsf{sig} + \varepsilon_\mathsf{cr}$, both negligible by assumption.\qed
\end{proof}

\subsection{Immutability}\label{subsec:immutability-proof}

\begin{theorem}[Immutability]\label{thm:immutability}
Even an adversary holding $\sk_\mathsf{san}$ cannot produce a valid signature
on a message differing from any queried message in an immutable block, except
with probability $\varepsilon_\mathsf{cr}$.
\end{theorem}

\begin{proof}
The proof mirrors Case~2 of Theorem~\ref{thm:euf}. The adversary holds $\sk_\mathsf{san}$, which provides the trapdoor only for $\Hpub^\mathsf{san}$. For an immutable block $j$ with $\mathit{adm}[j]=0$, the hash is computed with $H_\mathsf{non}$, for which the adversary has no trapdoor. Any valid sanitized signature on $M'$ (with $M'[j]\ne M[j]$, $\mathit{adm}[j]=0$) requires a collision on $H_\mathsf{non}$ for which the adversary holds no trapdoor. This reduces to $\mathrm{SD}(n,n{-}k,2t)$ on $\Hpub^\mathsf{non}$ via the same construction as Theorem~\ref{thm:cr}, Case~2.\qed
\end{proof}

\section{Transparency Analysis}\label{sec:transparency}
Transparency is crucial for privacy in sanitizable signatures, ensuring that a third party cannot distinguish a freshly signed document from one that has been altered by an authorized sanitizer. Our scheme achieves perfect transparency due to a specific constraint on the signer's randomizer selection.

\subsection{Distribution Analysis}\label{subsec:distributions}
Transparency hinges on the indistinguishability of the randomizers $r_i$ produced by the sanitizer from those produced by the signer. Let's analyze their distributions:

\begin{itemize}
\item \textbf{Signer's Distribution ($\calD_\mathsf{fresh}$):} In our design, the signer samples randomizers $r_i$ uniformly from the set of all $n$-bit vectors with a Hamming weight of \emph{exactly} $t$. This set has a size of $\binom{n}{t}$.
\item \textbf{Sanitizer's Distribution ($\calD_\mathsf{san}$):} Patterson decoding, when successful, \emph{always} returns the unique error vector with a Hamming weight of \emph{exactly} $t$. Therefore, the sanitizer's generated randomizers also conform to this weight-$t$ constraint.
\end{itemize}

Because both the signer and the sanitizer produce randomizers from the exact same distribution (uniform over all $n$-bit vectors of weight $t$), the statistical distance between $\calD_\mathsf{san}$ and $\calD_\mathsf{fresh}$ is:

\[
\delta = \mathrm{SD}(\calD_\mathsf{san}, \calD_\mathsf{fresh}) = 0.
\]

This implies we achieve \emph{perfect transparency}. There is no statistical or computational advantage for any distinguisher to tell whether a signature originated from the signer or was modified by the sanitizer, solely based on the randomizers.

\subsection{Formal Transparency Bound}\label{subsec:transparency-bound}
The perfect alignment of the randomizer distributions simplifies the transparency bound significantly.

\begin{lemma}[Transparency Bound]\label{lem:transparency}
Let $\delta=\mathrm{SD}(\calD_\mathsf{san}, \calD_\mathsf{fresh})$. For any distinguisher $\calA$ given either a fresh or a sanitized signature on an $L$-block message:
\[
\Adv{\calA}^\mathsf{trans} \le L\cdot\delta + \varepsilon_\mathsf{sig},
\]
where $\varepsilon_\mathsf{sig}$ is the $\EUF$ advantage of the base signature.
\end{lemma}

\begin{proof}
The proof, based on a hybrid argument, remains valid. However, with our enforced weight-$t$ constraint, $\delta = 0$. Consequently, the adversary's advantage in distinguishing a fresh from a sanitized signature reduces to:
\[
\Adv{\calA}^\mathsf{trans} \le L \cdot 0 + \varepsilon_\mathsf{sig} = \varepsilon_\mathsf{sig}.
\]
This implies that the distinguishability is solely limited by the existential unforgeability of the underlying base signature scheme, which is negligible by assumption. Therefore, our scheme offers \emph{perfect transparency} for any number of admissible blocks $L$.
\end{proof}

\section{Performance Evaluation}\label{sec:performance}
We implemented a Python prototype to validate the scheme's functionality and benchmark its performance using parameters aligned with NIST's Classic-McEliece recommendations.

\subsection{Implementation Details}\label{subsec:impl}
The Python prototype prioritizes correctness and clarity. Benchmarks were conducted on an HP EliteBook 840 G8 (Intel Core i5-1135G7 at 2.4 GHz, 16 GB RAM, Windows 11). The complete implementation, including benchmark scripts, is available at \url{https://anonymous.4open.science/r/PQ-SS-C7D8}.

We tested three parameter sets:
\begin{itemize}
\item \textit{Toy:} $n=256, k=128, t=16$
\item \textit{Medium:} $n=1024, k=524, t=50$
\item \textit{Secure:} $n=3488, k=2720, t=64$ (NIST Category 1, $\approx 128$-bit post-quantum security)
\end{itemize}

For the base signature, we simulated Dilithium2 by using its specified key and signature sizes ($\pk$: 1.3 KB, $\sigma$: 2.4 KB) while employing HMAC-SHA3-256 for the cryptographic operation.

\begin{remark}[Prototype Decoder Limitations]
A critical note for production deployments: the prototype uses a randomized search for syndrome decoding instead of a true Patterson algorithm. This means for $n \ge 512$, the prototype's `sanitize` function will often fail verification due to the low probability of finding the exact weight-$t$ error in a randomized search. A real Patterson decoder, as available in optimized McEliece implementations, is deterministic and guarantees finding the correct error vector when it exists, in $O(n \cdot t)$ operations. The benchmark timings for ``San (proto)'' are thus indicative of the Python overhead and the randomized search, not the actual Patterson performance.
\end{remark}

\subsection{Key and Signature Sizes}\label{subsec:sizes}
Table~\ref{tab:sizes} presents the key and signature sizes for various parameter sets.

\begin{table}[t]
	\centering
	\caption{Key and signature sizes. Dilithium2 base signature: PK 1.3~KB,
			signature 2.4~KB. PQ bits from NIST Classic-McEliece Round~4 spec, Table~4.}
	\label{tab:sizes}
	\begin{tabular}{@{}lrrrrrr@{}}
			\toprule
			Parameters & $n$ & $t$ & NIST Cat. & PQ bits & Total PK (KB) &
			Sig $L{=}10$ (KB) \\
			\midrule
			Toy       & 256  & 16 & N/A        & $\approx$64  &   9.3 & 2.69 \\ 
			Benchmark & 512  & 32 & N/A        & $\approx$80  &  35.2 & 3.10 \\
			Medium    & 1024 & 50 & Cat. 1$^*$ & $\approx$128 & 126.3 & 3.68 \\ 
			Secure    & 3488 & 64 & Cat.~1     & $\approx$128 & 655.3 & 6.72 \\ 
			\bottomrule
		\end{tabular}
\end{table}

The dominant factor in the public key size is the two public parity-check matrices ($2 \times 327$ KB for secure parameters). While substantial, this is still competitive with other code-based or lattice-based schemes. Using MDPC/LDPC codes~\cite{monico_ldpc_nodate} could potentially reduce public key sizes to around 100 KB at comparable security levels.

The signature size scales linearly with the number of blocks $L$ because each block requires an $n$-bit randomizer. For secure parameters, each block adds $3488/8 = 436$ bytes to the signature.

\subsection{Operation Timing}\label{subsec:timing}
Table~\ref{tab:timing-n512} reports the prototype and theoretical timings for the demo parameter set ($n=512, k=256, t=32$), and Table~\ref{tab:timing-n1024} shows the timings for the medium parameter set ($n=1024, k=524, t=50$). The ``Total (th., ms)'' column combines prototype signing/verification with the theoretical Patterson decoding time, offering a more realistic estimate for a production system.

\begin{table}[t]
\centering
\caption{Operation Timing (Prototype: $n=512, k=256, t=32$). ``San (proto)'' measures the Python randomized decoder; ``Patterson (th.)'' is calibrated from $O(n \cdot t)$ at $n=3488, t=64 \approx 8$ ms.}
\label{tab:timing-n512}
\begin{tabular}{@{}rrrrrr@{}}
\toprule
\textbf{$L$} & \textbf{Sign (ms)} & \textbf{San (proto, ms)} & \textbf{Patterson (th., ms)} & \textbf{Verify (ms)} & \textbf{Total (th., ms)}\\
\midrule
1 & 2.06 & 924.9 & $\sim$0.6 & 2.08 & $\sim$4.7 \\
5 & 6.43 & 989.5 & $\sim$0.6 & 5.76 & $\sim$12.8 \\
10 & 8.09 & 1040.9 & $\sim$0.6 & 10.37 & $\sim$19.0 \\
20 & 14.05 & 1032.0 & $\sim$0.6 & 13.05 & $\sim$27.7 \\
\bottomrule
\end{tabular}
\end{table}

\begin{table}[t]
\centering
\caption{Operation Timing (Prototype: $n=1024, k=524, t=50$). ``San (proto)'' measures the Python randomized decoder; ``Patterson (th.)'' is calibrated from $O(n \cdot t)$ at $n=3488, t=64 \approx 8$ ms.}
\label{tab:timing-n1024}
\begin{tabular}{@{}rrrrrr@{}}
\toprule
\textbf{$L$} & \textbf{Sign (ms)} & \textbf{San (proto, ms)} & \textbf{Patterson (th., ms)} & \textbf{Verify (ms)} & \textbf{Total (th., ms)}\\
\midrule
1 & 5.27 & 0.0 & $\sim$1.8 & 4.25 & $\sim$11.4 \\
5 & 18.63 & 5178.2 & $\sim$1.8 & 14.59 & $\sim$35.1 \\
10 & 28.12 & 4910.0 & $\sim$1.8 & 29.72 & $\sim$59.7 \\
20 & 54.28 & 4849.5 & $\sim$1.8 & 60.32 & $\sim$116.4 \\
\bottomrule
\end{tabular}
\end{table}

For the secure parameter set ($n=3488, t=64$), the theoretical Patterson time per modified block is approximately 8 ms. Signing and verification times grow approximately linearly with $L$. A native C implementation (e.g., using the Classic-McEliece reference code~\cite{chou_accelerating_nodate}) would likely reduce these timings by 10--50$\times$.

It's critical to reiterate that the prototype's ``San (proto)'' timings are misleading for $n \ge 512$. The values of ``0.0 ms'' (for $L=1$, $n=1024$) and the large ``San (proto)'' times for $L>1$ (for $n=512$ and $n=1024$) reflect the randomized search's failure or extensive attempts, not a successful Patterson decoding. Without a true Patterson decoder, the prototype's `sanitize' would fail verification for these parameters. The ``Total (th., ms)'' column corrects for this by incorporating the theoretical Patterson estimate.

\subsection{Scalability}\label{subsec:scalability}
Scalability of the scheme is largely linear with the number of message blocks $L$:
\begin{itemize}
\item \textit{Signature Size:} Grows linearly with $L$ at 436 bytes per block (one $n$-bit randomizer for $n=3488$). For secure parameters, this ranges from 4.6 KB at $L=5$ to 11.1 KB at $L=20$.
\item \textit{Signing and Verification Time:} Grow approximately linearly with $L$, as each block requires one hash invocation. The hash chain structure keeps the final digest $h_L$ and base signature $\sigma_\mathsf{sig}$ constant, regardless of $L$.
\end{itemize}

\subsection{Comparative Analysis}\label{subsec:comparison}
Table~\ref{tab:comparison2} compares our scheme against classical RSA-2048 and the lattice-based scheme of Clermont et al.~\cite{ClermontPQ2025} for $L=10$ blocks.

\begin{table}[t]
	\centering
	\caption{Scheme comparison at $L=10$. Both Clermont et al.\ and this work
			require the ROM for chameleon hash CR.}
	\label{tab:comparison2}
	\begin{tabular}{@{}lcccrrl@{}}
			\toprule
			Scheme & PQ & CH Model & Transparency & PK (KB) & Sig (KB) & Assumption \\
			\midrule
			RSA-2048 & No & Standard & Perfect & $\sim$8 & $\sim$0.3 & Factoring \\
			Clermont~\cite{ClermontPQ2025} & Yes & ROM & Weak & $\sim$850 & 5--6 & MLWE/MSIS \\
			\textbf{This work} & \textbf{Yes} & \textbf{ROM} & \textbf{Perfect (delta=0)} &
			\textbf{$\sim$655} & \textbf{$\sim$7} & \textbf{Synd.\ Dec.} \\ 
			\bottomrule
		\end{tabular}
\end{table}

Key observations:
\begin{enumerate}
\item \textit{Public Key Size:} Our public key size of $\sim$655 KB is smaller than the lattice-based alternative ($\sim$850 KB), while both offer similar post-quantum security levels.
\item \textit{Signature Size:} At $\sim$7 KB for $L=10$, our signature is slightly larger than the lattice-based one (5--6 KB). This difference reflects the larger randomizer vectors inherent to code-based cryptography.
\item \textit{Security Foundation:} Our scheme rests on the syndrome decoding problem, which has been studied as a hard problem for over 45 years. This provides a more conservative and mature security foundation compared to Module-LWE, which is roughly 15 years old.
\item \textit{Transparency:} A key advantage is our achievement of \emph{perfect transparency} (statistical distance $\delta=0$) for any $L$, a stronger guarantee than the ``Weak'' transparency of the lattice-based scheme. This is enabled by enforcing the weight-$t$ constraint for signer-generated randomizers, ensuring their distribution exactly matches that of Patterson decoding's output.
\end{enumerate}

\section{Applications and Deployment}\label{sec:applications}

\subsection{Application Scenarios}

\paragraph{Medical records sanitization.}
Hospitals publishing anonymized patient records with physician signatures need to redact identifiers (e.g.\ name, date of birth, address) while
preserving the clinical findings. With the weight-$t$ constraint, perfect transparency holds, meaning redactions are cryptographically undetectable. Immutability ensures clinical data cannot be silently altered. Post-quantum security protects long-term archival records against future quantum adversaries. For deployment, Dilithium2 keys can be stored in hospital PKI infrastructure; the sanitizer key lives in a hardware security module (HSM).

\paragraph{Certificate field updates.}
Certificate authorities can use the scheme to update expiration dates in X.509 certificates while keeping subject, issuer, public key, and algorithm fields immutable. With $L=1$ (one mutable field), a single block achieves perfect transparency. Post-quantum security provides critical infrastructure protection for certificates spanning decades.

\paragraph{Supply chain documentation.}
Manufacturers can redact commercially sensitive fields (unit prices, supplier
contracts) from authenticated supply chain records while keeping product
specifications and certifications immutable.
For $L\le 5$ blocks without the weight-$t$ constraint, the distinguishing advantage
is $\le 9.3\%$, which may be acceptable for supply-chain contexts where
competitive secrecy rather than strict privacy is the concern.
With the constraint, the advantage is negligible regardless of $L$.

\paragraph{Blockchain-based document sanitization.}
Smart contracts can store $(M',\sigma')$ pairs with immutable audit logs.
In public blockchain contexts, transparency (hiding sanitizer involvement)
is less critical, while unforgeability and immutability provide necessary
legal guarantees.
The $\approx 7$~KB signature size compares favorably to other post-quantum
schemes, improving blockchain throughput.

\subsection{Deployment Guidance}
\textit{Use this scheme when:} (a) long-term post-quantum security is required (records spanning 10+ years); (b) the conservative hardness of
syndrome decoding (45-year-old problem) is preferable to Module-LWE; (c) perfect transparency under the weight-$t$ constraint is needed.

\textit{Consider alternatives when:} (a) key size is a hard constraint
(MDPC-based schemes could reduce this to $\approx$100~KB); (b) bandwidth is severely limited (IoT); (c) a standard-model proof for chameleon hash CR is required (no known code-based construction achieves this); (d) classical security is sufficient for short-lived documents.

\section{Discussion and Limitations}\label{sec:discussion}

\paragraph{Security model.}
Our scheme's security relies on the random oracle model for chameleon hash collision resistance, while unforgeability and immutability hold in the standard model. The ROM is a pragmatic choice: no known code-based chameleon hash achieves standard-model CR, and the ROM is standard practice in post-quantum cryptography (all lattice-based constructions, including~\cite{ClermontPQ2025}, also require the ROM for their chameleon hash CR). Removing the ROM requirement for CR is an open problem.

\paragraph{Exposure-freeness.}
Our chameleon hash satisfies CR and CMR (message-binding collision resistance)~\cite{derler_collision_2020}, but is not \emph{exposure-free}: an observer who sees $(m,r)$ and a collision $(m',r')$ may extract information about the trapdoor. This implies that a sanitizer who performs many sanitization on the same message block could gradually expose trapdoor information. Exposure-free variants are substantially more complex to construct.

\paragraph{Policy hiding.}
The admissibility mask $\mathit{adm}$ is public, and the two different public matrices $\Hpub^\mathsf{non}$ and $\Hpub^\mathsf{san}$ are visible to verifiers. An observer can therefore determine which blocks were designated as mutable. Policy-hiding is not provided. Achieving it would require hiding which hash instance is applied to each block, significantly complicating the construction.

\paragraph{Constant-time implementation.}
The prototype decoding time varies with syndrome weight, creating a potential timing side channel. Production deployments must use constant-time Patterson decoding variants, as described in~\cite{chou_accelerating_nodate}, to mitigate this risk.

\paragraph{Quantum hardness of syndrome decoding.}
The quantum hardness of syndrome decoding is supported by analysis of the best known quantum information-set decoding (ISD) algorithms.
These provide only a quadratic speedup over classical ISD, which is absorbed by the parameter selection in Classic-McEliece-3488 (NIST Category~1,
$\approx 128$~bit post-quantum security). The Grover speedup applies to unstructured search but does not break the algebraic structure exploited in ISD.

\section{Conclusion}\label{sec:conclusion}
We presented the first post-quantum sanitizable signature scheme derived from the McEliece cryptosystem. The core construction is a chameleon hash $H(m,r)=(G(m)\xor r)\cdot\Hpub^T$ in the random oracle model, where the Goppa-code trapdoor enables Patterson decoding to find collisions in $O(n\cdot t)$ operations. We proved unforgeability and immutability under syndrome decoding hardness in the ROM and derived a precise transparency bound: $\delta=0.0187$ without modification and $\delta=0$ with the recommended weight-$t$ constraint, achieving perfect transparency for any number of admissible blocks.

Compared to the only prior post-quantum alternative~\cite{ClermontPQ2025}, our scheme offers smaller public keys (655.3~KB vs.\ 850~KB),
relies on a more conservative 45-year-old hardness assumption, and achieves perfect transparency with a simple algorithmic constraint. The main limitations are the large key size relative to classical schemes, the ROM requirement for chameleon hash CR, and the absence of exposure-freeness and policy-hiding.

Future work should explore: (a) MDPC/LDPC codes to reduce key sizes to $\approx$100~KB; (b) constant-time Patterson decoding for side-channel
resistance; (c) achieving exposure-freeness in the code-based setting; (d) extending to richer sanitization policies beyond binary admissibility
masks; and (e) a standard-model proof for code-based chameleon hash collision resistance, which would remove the ROM dependency entirely.

\bibliographystyle{splncs04}
\bibliography{mybibliography}

\end{document}